\documentclass[a4paper,10pt]{article}
\usepackage[utf8x]{inputenc}
\usepackage{amsmath,amsthm,amssymb,epsfig,fullpage}

\newtheorem{theorem}{Theorem}

\newtheorem{remark}{Remark}

\newtheorem{definition}{Definition}

\numberwithin{theorem}{section}
\numberwithin{lemma}{section}
\numberwithin{proposition}{section}
\numberwithin{equation}{section}
\numberwithin{remark}{section}
\numberwithin{remarks}{section}
\numberwithin{definition}{section}
\numberwithin{scheme}{section}
\numberwithin{example}{section}
\numberwithin{corollary}{section}
\usepackage{colortbl}

\newcommand\del[1]{}

\title{Hashimoto transform for stochastic Landau-Lifshitz-Gilbert equation}
\author{}

\begin{document}

\maketitle
\section{Introduction}

It is well known that the Heisenberg ferromagnet equation and the nonlinear Schr\"{o}dinger equation are equivalent (\cite{FT2007}, \cite{H1972}, \cite{NSVZ2007} and references therein). In this paper we will show that Hashimoto transformation is applicable to the one dimensional stochastic Landau-Lifshitz-Gilbert (LLG) equation (\cite{BBNP2013}, \cite{BGJ2013} and references therein) and transforms it to the stochastic nonlinear generalized heat equation with nonlocal (in space) interaction. We will start with the case of deterministic 1D LLG equation
\begin{equation}\label{eqn:heatflowharmmaps_d0}
\mathbf{u}_t=\beta\mathbf{u}\times \mathbf{u}_{xx}-\alpha\mathbf{u}\times(\mathbf{u}\times \mathbf{u}_{xx}),x\in \mathbb{R},t\geq 0.
\end{equation}
and prove in the Theorem \ref{thm:deterministic_case} that Hashimoto transform $q=|\mathbf{u}_{x}|e^{i\int\limits_{-\infty}^{\cdot}\frac{<\mathbf{u}\times\mathbf{u}_x,\mathbf{u}_{xx}>}{|\mathbf{u}_x|^2}\,dy}$ of a smooth solution satisfies to the following generalized heat equation
\begin{align}\label{eqn:heatflowharmmap_0}
\partial_t q &= \alpha\Big[q_{xx}+\frac{q}{2}\int\limits_{-\infty}^x
\left(q_x\overline{q}-q\overline{q}_x\right)\,dy\Big]+i\beta(q_{xx}+\frac{|q|^2}{2}q)\\
&\equiv (\alpha+i\beta)\Big[q_{xx}+\frac12 q|q|^2\Big]-\alpha q\int\limits_{-\infty}^x q\,\overline{q}_x\,dy.
\end{align}

Then we will consider the stochastic case and show that the following stochastic nonlinear 
heat equation 
\begin{align}
d_t\mathbf{q} &= \left((\alpha+i\beta)\mathbf{q}_{xx}+i\beta\frac{\mathbf{q}|\mathbf{q}|^2}{2}+\frac{\alpha\mathbf{q}}{2}\int\limits_{a}^x(\mathbf{q}_x\overline{\mathbf{q}}
-\overline{\mathbf{q}}_x\mathbf{q})\,dy\right)\,dt\nonumber\\
&\quad+d\partial_x(W^1+iW^2)-i\mathbf{q}\int\limits_a^x q^2\circ d_tW^1-q^1\circ d_tW^2,\label{eqn:Nonlinear_Stoch_Heat_0}
\end{align}
allows one to construct (assuming existence and certain smoothness of the solution) weak solution of stochastic LLG equation (Theorem \ref{thm:main})
\begin{equation}
d \mathbf{u}=(\beta\mathbf{u}\times\mathbf{u}_{xx}-
\alpha\mathbf{u}\times(\mathbf{u}\times\mathbf{u}_{xx}))\,dt+\mathbf{u}\times \circ d\widetilde{\bf W}_t,\label{eqn:SLLGtarget_0}
\end{equation}
where $\{\widetilde{\bf W}_t\}_{t\geq 0}$ is 3D space-time white noise iff we formally assume that $W^1,W^2$ are independent space-time white noises (remark \ref{rem:Space_time_White_noise}).
 We will see that nonlinear stochastic heat equation \eqref{eqn:Nonlinear_Stoch_Heat_0} is the compatibility condition for the auxiliary system
of linear equations \eqref{eqn:time_evolution}-\eqref{eqn:space_evolution} (where $\mathbf{p}, C(\mathbf{q})$ are defined by identities \eqref{eqn:internal_structure}). The stochastic LLG equation is the consequence of the auxiliary system \eqref{eqn:time_evolution},\eqref{eqn:space_evolution},\eqref{eqn:internal_structure}.

\section{Equivalence of LLG equation and generalized heat equation}
\begin{definition}\label{def:HashimotoTransform}
Let $\mathbf{u}:\mathbb{R}\to \mathbb{S}^2$ be a smooth function.
Define 
\begin{equation}
\Theta(\mathbf{u}):=|\mathbf{u}_x|,\quad\eta(\mathbf{u}):=\frac{<\mathbf{u}\times\mathbf{u}_x,\mathbf{u}_{xx}>}{|\mathbf{u}_x|^2}.
\end{equation}
Now the Hashimoto transform $\mathbf{u}\mapsto \mathcal{H}(u)\in C^{\infty}(\mathbb{R},\mathbb{C})$ is defined as follows
\begin{equation}\label{eqn:Hashimoto_transform}
\mathcal{H}(\mathbf{u}):=\Theta(\mathbf{u}) e^{i\int\limits_{-\infty}^{x}\eta(\mathbf{u})(y)\,dy}.
\end{equation}
\end{definition}
\begin{remark}
We can connect with $\mathbf{u}$ orthonormal basis $\{\mathbf{u},\frac{\mathbf{u}_x}{|\mathbf{u}_x|},\frac{\mathbf{u}\times\mathbf{u}_x}{|\mathbf{u}_x|}\}$. Then $\Theta(\mathbf{u})$ is a curvature of the basis and $\eta(\mathbf{u})$ is a torsion of the basis.
\end{remark}
We will need following identities
\begin{equation}\label{eqn:Hashimoto_transform_inverse}
\Theta(\mathbf{u})=|\mathcal{H}(\mathbf{u})|,\quad\eta(\mathbf{u})=i\frac{\mathcal{H}(\mathbf{u})\overline{\mathcal{H}(\mathbf{u})}_x-\mathcal{H}_x(\mathbf{u})\overline{\mathcal{H}(\mathbf{u})}}{2|\mathcal{H}(\mathbf{u})|^2}.
\end{equation}
\begin{theorem}\label{thm:deterministic_case}
Let $\mathbf{u}:[0,\infty)\times\mathbb{R}\to \mathbb{S}^2$ be a smooth solution of the Landau-Lifshitz-Gilbert equation
\begin{equation}\label{eqn:heatflowharmmaps_d1}
\mathbf{u}_t=\beta\mathbf{u}\times \mathbf{u}_{xx}-\alpha\mathbf{u}\times(\mathbf{u}\times \mathbf{u}_{xx}),x\in \mathbb{R},t\geq 0.
\end{equation}
Then its Hashimoto transform $q=\mathcal{H}(\mathbf{u}(t,\cdot)),t\in[0,\infty)$, $q:[0,\infty)\times\mathbb{R}\to\mathbb{C}$ is a smooth solution of the following equation
\begin{align}\label{eqn:heatflowharmmap_4}
\partial_t q &= \alpha\Big[q_{xx}+\frac{q}{2}\int\limits_{-\infty}^x
\left(q_x\overline{q}-q\overline{q}_x\right)\,dy\Big]+i\beta(q_{xx}+\frac{|q|^2}{2}q)\\
&\equiv (\alpha+i\beta)\Big[q_{xx}+\frac12 q|q|^2\Big]-\alpha q\int\limits_{-\infty}^x q\,\overline{q}_x\,dy.
\end{align}
\end{theorem}
\begin{proof}
Our proof will be divided in two steps. First, we will deduce equation for the pair
$(\Theta,\eta)$ (we will omit argument $\mathbf{u}$ from now on). In the second step, we will
deduce equation for $q$.

{\bf Step 1:}

We have by elementary calculations and identities 
\[
<\mathbf{u}_x,\mathbf{u}>=\frac12\partial_x(|\mathbf{u}|^2)=0, |\mathbf{u}_x|^2=\partial_x(<\mathbf{u}_x,\mathbf{u}>)-<\mathbf{u}_{xx},\mathbf{u}>=-<\mathbf{u}_{xx},\mathbf{u}>,
\]
that
\begin{align*}
\Theta' &=\frac{<\mathbf{u}_x,\mathbf{u}_{tx}>}{\Theta}\\
&=\frac{<\beta\mathbf{u}\times\mathbf{u}_{xxx}+\beta\mathbf{u}_x\times \mathbf{u}_{xx}+\alpha \mathbf{u}_{xxx}+\alpha|\mathbf{u}_x|^2\mathbf{u}_x+\alpha \mathbf{u}\partial_x(|\mathbf{u}_x|^2),\mathbf{u}_x>}{\Theta}\\
&=\frac{<\beta\mathbf{u}\times\mathbf{u}_{xxx}+\alpha \mathbf{u}_{xxx}+\alpha|\mathbf{u}_x|^2\mathbf{u}_x,\mathbf{u}_x>}{\Theta}\\
&=\frac{\beta<\mathbf{u}_x,\mathbf{u}\times\mathbf{u}_{xxx}>+\alpha<\mathbf{u}_{xxx},\mathbf{u}_x>+\alpha |\mathbf{u}_x|^4}{\Theta}\\
&=\frac{\beta\partial_x (<\mathbf{u}_x\times\mathbf{u},\mathbf{u}_{xx}>)+\alpha\left[\partial_x(<\mathbf{u}_{xx},\mathbf{u}_x>)-|\mathbf{u}_{xx}|^2\right]+\alpha |<\mathbf{u}_{xx},\mathbf{u}>|^2}{\Theta}\\
&=\frac{-\beta\partial_x(\eta\Theta^2)+\alpha\left[\partial_{xx}^2\left(\frac{\Theta^2}{2}\right)-|\mathbf{u}_{xx}|^2+|<\mathbf{u}_{xx},\mathbf{u}>|^2\right]}{\Theta}\\
&=-\beta\eta_x\Theta-2\beta\Theta_x\eta+\alpha \Theta_{xx}+\alpha\frac{|\Theta_x|^2}{\Theta}-\alpha\frac{|\mathbf{u}_{xx}|^2-|<\mathbf{u}_{xx},\mathbf{u}>|^2}{\Theta}.
\end{align*}
Consequently, by elementary identity  $|\mathbf{a}|^2|\mathbf{b}|^2=|\mathbf{a}\times \mathbf{b}|^2+|<\mathbf{a},\mathbf{b}>|^2,\mathbf{a},\mathbf{b}\in\mathbb{R}^3$ we have that
\begin{equation}\label{eqn:Theta_1}
\Theta'=\alpha\Theta_{xx}+\alpha\frac{|\Theta_{x}|^2}{\Theta}
-\alpha\frac{|\mathbf{u}\times\mathbf{u}_{xx}|_{\mathbf{R}^3}^2}{\Theta}-
\beta\eta_x\Theta-2\beta\Theta_x\eta.
\end{equation}
Denote $\mathbf{v}:=\mathbf{u}\times\mathbf{u}_x$. 
Now 
\begin{align}
&\frac{|\Theta_{x}|^2-|\mathbf{u}\times\mathbf{u}_{xx}|_{\mathbf{R}^3}^2}{\Theta}=
\frac{\frac{|\partial_x(\Theta^2)|^2}{4\Theta^2}-|\mathbf{v}_x|^2}{|\mathbf{v}|}\nonumber\\
&=\frac{\frac{|\partial_x(|\mathbf{v}|^2)|^2}{4|\mathbf{v}|^2}-|\mathbf{v}_x|^2}{|\mathbf{v}|}
=-\frac{|\mathbf{v}\times\mathbf{v}_x|^2}{|\mathbf{v}|^3}\nonumber\\
&=-\frac{|<\mathbf{u}\times\mathbf{u}_x,\mathbf{u}_{xx}>|^2}{|\mathbf{u}_x|^3}
=-\eta^2\Theta.\label{eqn:ThetaV}
\end{align}
Hence by equation \eqref{eqn:Theta_1} and identity \eqref{eqn:Theta_1} we can deduce that
\begin{equation}\label{eqn:Theta}
\Theta'=\alpha(\partial_{xx}^2-\eta^2)\Theta-\beta\eta_x\Theta-2\beta\Theta_x\eta.
\end{equation}
It remains to deduce equation for $\eta$. We have
\begin{align}\label{eqn:eta_aux_0}
\eta' &=\frac{-2\alpha(\partial_{xx}^2-\eta^2)\Theta-\beta\eta_x\Theta-2\beta\Theta_x\eta}{\Theta^3}\eta\Theta^2+\frac{1}{\Theta^2}\frac{d}{dt}<\mathbf{u}\times\mathbf{u}_x,\mathbf{u}_{xx}>\\
&=-2\alpha\eta\frac{\Theta_{xx}}{\Theta}+2\alpha\eta^3-\beta\eta_x-2\beta\eta\frac{\Theta_x}{\Theta}+\frac{1}{\Theta^2}
\big(<\mathbf{u}',\mathbf{u}_x\times\mathbf{u}_{xx}>+<\mathbf{u}\times\mathbf{u}_x',\mathbf{u}_{xx}>+<\mathbf{u}\times\mathbf{u}_x,\mathbf{u}_{xx}'>\big)\nonumber\\
&=-2\alpha\eta\frac{\Theta_{xx}}{\Theta}+2\alpha\eta^3-\beta\eta_x-2\beta\eta\frac{\Theta_x}{\Theta}+\frac{1}{\Theta^2}(A+B+C)\nonumber
\end{align}
Now let us calculate $A,B$ and $C$. We will deal separately with the terms proportional to $\alpha$ and $\beta$. We have $A=A_{\alpha}+A_{\beta}$, where
\begin{equation}\label{eqn:eta_aux_1}
A_{\alpha}=\alpha<\mathbf{u}_{xx}+|\mathbf{u}_x|^2\mathbf{u},\mathbf{u}_x\times\mathbf{u}_{xx}>=
\alpha\eta\Theta^4,\nonumber
\end{equation}
\begin{align}\label{eqn:eta_aux_1_beta}
A_{\beta}&=\beta<\mathbf{u}\times\mathbf{u}_{xx},
\mathbf{u}_x\times\mathbf{u}_{xx}>\\
&=-\beta<\mathbf{u},\mathbf{u}_{xx}><\mathbf{u}_x,\mathbf{u}_{xx}>=\beta\Theta^2\frac{\partial_x(\Theta^2)}{2}=\beta\Theta^3\Theta_x.\nonumber
\end{align}
We have $B=B_{\alpha}+B_{\beta}$, where
\begin{align}\label{eqn:eta_aux_2}
B_{\alpha}&=-\alpha <\mathbf{u}_{xxx},\mathbf{u}\times\mathbf{u}_{xx}>-\alpha|\mathbf{u}_x|^2<\mathbf{u}_x,\mathbf{u}\times \mathbf{u}_{xx}>\\
&=\alpha\eta\Theta^4-\alpha <\mathbf{u}_{xxx},\mathbf{u}\times\mathbf{u}_{xx}>.\nonumber
\end{align}
To calculate $B_{\beta}$ we will need the following auxiliary identity.
\begin{equation}\label{eqn:eta_aux_8}
|\mathbf{u}_{xx}|^2=\Theta^4+|\Theta_x|^2+\eta^2\Theta^2.
\end{equation}
Indeed, expanding $|\mathbf{u}_{xx}|^2$ in the orthonormal basis $\{\mathbf{u},\frac{\mathbf{u}_x}{|\mathbf{u}_x|},\frac{\mathbf{u}\times\mathbf{u}_x}{|\mathbf{u}_x|}\}$ we get
\begin{align*}
|\mathbf{u}_{xx}|^2&=|<\mathbf{u}_{xx},\mathbf{u}>|^2+|<\mathbf{u}_{xx},\frac{\mathbf{u}_x}{|\mathbf{u}_x|}>|^2+|<\mathbf{u}_{xx},\frac{\mathbf{u}\times\mathbf{u}_x}{|\mathbf{u}_x|}>|^2\\
&=\Theta^4+\frac{1}{\Theta^2}|\partial_x\left(\frac{\Theta^2}{2}\right)|^2+\eta^2\Theta^2\\
&=\Theta^4+|\Theta_x|^2+\eta^2\Theta^2.
\end{align*}
Hence by elementary calculations and identity \eqref{eqn:eta_aux_8} we get
\begin{align}\label{eqn:eta_aux_2_beta}
B_{\beta}&=\beta<\mathbf{u}\times (\mathbf{u}\times\mathbf{u}_{xxx}+\mathbf{u}_x\times\mathbf{u}_{xx}),\mathbf{u}_{xx}>\\
&=\beta<\mathbf{u}<\mathbf{u},\mathbf{u}_{xxx}>-\mathbf{u}_{xxx}+\mathbf{u}_x<\mathbf{u},\mathbf{u}_{xx}>, \mathbf{u}_{xx}>\nonumber\\
&=\beta<\mathbf{u},\mathbf{u}_{xx}><\mathbf{u},\mathbf{u}_{xxx}>-\beta<\mathbf{u}_{xx},\mathbf{u}_{xxx}>+\beta<\mathbf{u},\mathbf{u}_{xx}><\mathbf{u}_x,\mathbf{u}_{xx}>\nonumber\\
&=\beta<\mathbf{u},\mathbf{u}_{xx}>\partial_x(<\mathbf{u},\mathbf{u}_{xx}>)-<\mathbf{u}_{xx},\mathbf{u}_{xxx}>=-\beta\Theta^2\partial_x(-\Theta^2)-\beta\frac12\partial_x(|\mathbf{u}_{xx}|^2)\nonumber\\
&=2\beta\Theta^3\Theta_x-\beta\frac12\partial_x(\Theta^4+|\Theta_x|^2+\eta^2\Theta^2)\nonumber\\
&=-\beta\Theta_x\Theta_{xx}-\beta\eta\eta_x\Theta^2-\beta\Theta\Theta_x\eta^2.\nonumber
\end{align}
We have $C=C_{\alpha}+C_{\beta}$, where
\begin{equation}\label{eqn:eta_aux_3}
C_{\alpha}=\alpha<\mathbf{u}\times\mathbf{u}_x,\mathbf{u}_{xxxx}+|\mathbf{u}_x|^2\mathbf{u}_{xx}>=
\alpha\eta\Theta^4+\alpha<\mathbf{u}\times\mathbf{u}_x,\mathbf{u}_{xxxx}>.
\end{equation}
\begin{align}
C_{\beta}&=\beta<\mathbf{u}\times\mathbf{u}_x,
\mathbf{u}\times\mathbf{u}_{xxxx}+2\mathbf{u}_x\times\mathbf{u}_{xxx}>\nonumber\\
&=\beta <\mathbf{u}_x,\mathbf{u}_{xxxx}>-2\beta\Theta^2 <\mathbf{u},\mathbf{u}_{xxx}>
\end{align}
Now we can notice that 
\[
<\mathbf{u},\mathbf{u}_{xxx}>=\partial_x(<\mathbf{u},\mathbf{u}_{xx}>)-<\mathbf{u}_x,\mathbf{u}_{xx}>=\partial_x(-\Theta^2)-\partial_x(\frac{\Theta^2}{2})=-3\Theta\Theta_x.
\]
Furthermore,
\begin{align*}
<\mathbf{u}_x,\mathbf{u}_{xxxx}>&=\partial_x (<\mathbf{u}_x,\mathbf{u}_{xxx}>)-<\mathbf{u}_{xx},\mathbf{u}_{xxx}>\\
&=\partial_x(\partial_x(<\mathbf{u}_x,\mathbf{u}_{xx}>)-|\mathbf{u}_{xx}|^2)-\frac12\partial_x (|\mathbf{u}_{xx}|^2)\\
&= -\frac32\partial_x(|\mathbf{u}_{xx}|^2)+\frac12\partial^3_{xxx}(|\mathbf{u}_x|^2)=-\frac32\partial_x(\Theta^4+|\Theta_x|^2+\eta^2\Theta^2)+\frac12\partial^3_{xxx}(\Theta^2).
\end{align*}
Hence,
\[
C_{\beta}=\beta(\Theta\Theta_{xxx}-3\eta\eta_x\Theta^2-3\Theta\Theta_x\eta^2).
\]
Thus
\begin{equation}\label{eqn:eta_aux_4}
A_{\alpha}+B_{\alpha}+C_{\alpha}=3\alpha\eta\Theta^4+\alpha \big(<\mathbf{u}\times\mathbf{u}_x,\mathbf{u}_{xxxx}>+
<\mathbf{u}_{xx}\times\mathbf{u},\mathbf{u}_{xxx}>\big),
\end{equation}
and
\begin{equation}\label{eqn:eta_aux_4_beta}
A_{\beta}+B_{\beta}+C_{\beta}=\beta(\Theta\Theta_{xxx}-\Theta_x\Theta_{xx}+\Theta^3\Theta_x-
4\eta\eta_x\Theta^2-4\Theta\Theta_x\eta^2).
\end{equation}
Since
\begin{equation}\label{eqn:eta_aux_5}
\partial_{xx}^2(\eta\Theta^2)=\partial_x(<\mathbf{u}\times\mathbf{u}_x,\mathbf{u}_{xxx}>)
=\big(<\mathbf{u}\times\mathbf{u}_x,\mathbf{u}_{xxxx}>+
<\mathbf{u}\times\mathbf{u}_{xx},\mathbf{u}_{xxx}>\big),
\end{equation}
we deduce that
\begin{equation}\label{eqn:eta_aux_6}
A_{\alpha}+B_{\alpha}+C_{\alpha}=3\alpha\eta\Theta^4+\alpha\partial_{xx}^2(\eta\Theta^2)+2\alpha
<\mathbf{u}_{xx}\times\mathbf{u},\mathbf{u}_{xxx}>.
\end{equation}
Thus, it remains to calculate $<\mathbf{u}_{xx}\times\mathbf{u},\mathbf{u}_{xxx}>$. Expanding this quantity in orthonormal basis $\{\mathbf{u},\frac{\mathbf{u}_x}{|\mathbf{u}_x|},\frac{\mathbf{u}\times\mathbf{u}_x}{|\mathbf{u}_x|}\}$ we get
\begin{align}\label{eqn:eta_aux_7}
<\mathbf{u}_{xx}\times\mathbf{u},\mathbf{u}_{xxx}>&=
<\mathbf{u}_xx\times\mathbf{u},\mathbf{u}><\mathbf{u},\mathbf{u}_{xxx}>+
<\mathbf{u}_{xx}\times\mathbf{u},\frac{\mathbf{u}_x}{|\mathbf{u}_x|}><\frac{\mathbf{u}_x}{|\mathbf{u}_x|},\mathbf{u}_{xxx}>\\
&+<\mathbf{u}_{xx}\times\mathbf{u},\frac{\mathbf{u}\times\mathbf{u}_x}{|\mathbf{u}_x|}><\frac{\mathbf{u}\times\mathbf{u}_x}{|\mathbf{u}_x|},\mathbf{u}_{xxx}>\nonumber\\
&=0+\eta <\mathbf{u}_x,\mathbf{u}_{xxx}>-\frac{<\mathbf{u}_x,\mathbf{u}_{xx}>}{|\mathbf{u}_x|^2}
<\mathbf{u}\times\mathbf{u}_x,\mathbf{u}_{xxx}>\nonumber\\
&=\eta <\mathbf{u}_x,\mathbf{u}_{xxx}>-\frac{<\mathbf{u}_x,\mathbf{u}_{xx}>}{|\mathbf{u}_x|^2}
\partial_x(\eta\Theta^2)\nonumber\\
&=\eta\left[\partial_x(<\mathbf{u}_x,\mathbf{u}_{xx}>)-|\mathbf{u}_{xx}|^2\right]-
\frac{1}{\Theta^2}\partial_x(\eta\Theta^2)\partial_x(\frac{\Theta^2}{2})\nonumber\\
&=\eta\partial_{xx}\left(\frac{\Theta^2}{2}\right)-\eta |\mathbf{u}_{xx}|^2
-\frac{\Theta_x}{\Theta}\big(\Theta^2\eta+2\eta\Theta\Theta_x\big).\nonumber
\end{align}

Combining \eqref{eqn:eta_aux_7} and \eqref{eqn:eta_aux_8} we deduce that 
\begin{equation}\label{eqn:eta_aux_9}
<\mathbf{u}_{xx}\times\mathbf{u},\mathbf{u}_{xxx}>=\eta\Theta\Theta_{xx}+\eta\Theta_x^2-
\eta\Theta^4-\eta\Theta_x^2-\eta^3\Theta^2-\Theta\Theta_x\eta_x-2\eta|\Theta_x|^2.
\end{equation}
Hence, combining \eqref{eqn:eta_aux_6} and \eqref{eqn:eta_aux_9} we conclude that
\begin{equation}\label{eqn:eta_aux_10}
A_{\alpha}+B_{\alpha}+C_{\alpha}=\alpha\eta\Theta^4+\alpha\eta_{xx}\Theta^2+2\alpha\eta_x\Theta\Theta_x-
2\alpha\eta|\Theta_x|^2+4\alpha\eta\Theta\Theta_{xx}-2\alpha\eta^3\Theta^2.
\end{equation}
Finally, by identities \eqref{eqn:eta_aux_0}, \eqref{eqn:eta_aux_4_beta} and \eqref{eqn:eta_aux_10} 
we can conclude that 
\begin{equation}
\eta'=\alpha\eta_{xx}+2\alpha\left(\eta\frac{\Theta_x}{\Theta}\right)_x+\alpha\eta\Theta^2
+\beta(\frac{\Theta_{xx}}{\Theta}+\frac{\Theta^2}{2}-\eta^2)_x.
\end{equation}

Hence, we conclude that from LLG equation \eqref{eqn:heatflowharmmaps_d1} follows that
\begin{align}\label{eqn:heatflowharmmap_2}
\Theta' &= \alpha(\partial_{xx}^2-\eta^2)\Theta-\beta(\eta_x\Theta+2\Theta_x\eta),\\
\eta' &= \alpha\eta_{xx}+2\alpha\left(\eta\frac{\Theta_x}{\Theta}\right)_x+\alpha\eta\Theta^2+
\beta(\frac{\Theta_{xx}}{\Theta}+\frac{\Theta^2}{2}-\eta^2)_x,\nonumber
\end{align}
where $\Theta=|\mathbf{u}_x|$, $\eta=\frac{<\mathbf{u}\times\mathbf{u}_x,\mathbf{u}_{xx}>}{|\mathbf{u}_x|^2}$.

{\bf Step 2:}

By definition of $q$ we have,
\[
\partial_t q=q\left(\frac{\Theta'}{\Theta}+i\int\limits_{-\infty}^x\eta'dy\right).
\]
By identities \eqref{eqn:heatflowharmmap_2} we deduce that
\begin{align}\label{eqn:heatflowharmmap_3}
\partial_t q&=\alpha q\Big(\frac{\Theta_{xx}}{\Theta}-\eta^2+i\big(\eta_x+2\frac{\eta\Theta_x}{\Theta}\big)+i\int\limits_{\cdot}^x\eta\Theta^2\,dy\Big)+\beta q(
-(\eta_x+2\frac{\eta\Theta_x}{\Theta}\big)+i(\frac{\Theta_{xx}}{\Theta}-\eta^2+\frac{\Theta^2}{2}))\\
&=\alpha \Big(P+Q+R\Big)+\beta S
\end{align}
Now we will calculate these four terms:
\begin{align}
R&=i q\int\limits_{\cdot}^x\eta\Theta^2\,dy=\frac{q}{2}\int\limits_{-\infty}^x
\left(q_x\overline{q}-q\overline{q}_x\right)\,dy\nonumber\\
&=\frac12 q|q|^2-q\int\limits_{-\infty}^x q\,d\overline{q}\label{eqn:aux_Cterm}
\end{align}
To calculate terms $P$ and $Q$ we will need following auxiliary identities. They can be deduced immediately by elementary calculations.
\begin{align}\label{eqn:Hashimoto_transform_inverse_2}
\Theta &= |q|,\quad\eta = i\frac{q\overline{q}_x-q_x\overline{q}}{2|q|^2}\\
\Theta_x &= \frac{q_x\overline{q}+q\overline{q}_x}{2|q|},\nonumber\\
\Theta_{xx} &= -\frac{(q_x\overline{q}+q\overline{q}_x)^2}{4|q|^3}+\frac{q_{xx}\overline{q}+q\overline{q}_{xx}+2q_x\overline{q}_x}{2|q|},\nonumber\\
\eta_x &= \frac{i}{2}\frac{|q|^2 q\overline{q}_{xx}-q_{xx}\overline{q}|q|^2-q^2\overline{q}_x^2+q_x^2\overline{q}^2}{|q|^4}.\nonumber
\end{align}
Consequently,
\begin{equation}
P=q\Big(\frac{\Theta_{xx}}{\Theta}-\eta^2\Big)=\frac{q_{xx}}{2}+\overline{q}_{xx}\frac{q^2}{2|q|^2},\label{eqn:aux_Aterm}
\end{equation}
\begin{equation}\label{eqn:aux_Bterm}
Q=i q\eta_x+2i q\eta\frac{\Theta_x}{\Theta}=
\frac{q_{xx}}{2}-\overline{q}_{xx}\frac{q^2}{2|q|^2}.
\end{equation}
Consequently,
\begin{equation}\label{eqn:aux_Dterm}
S=i(q_{xx}+\frac{|q|^2}{2}q).
\end{equation}

Thus, by identities \eqref{eqn:heatflowharmmap_3}, \eqref{eqn:aux_Aterm}, 
\eqref{eqn:aux_Bterm}, \eqref{eqn:aux_Cterm} and \eqref{eqn:aux_Dterm} we can conclude that 
the equation \eqref{eqn:heatflowharmmap_4} holds.

\end{proof}

Now we are going to show that a solution of the nonlinear stochastic heat equation can be used to construct the weak solution of stochastic Landau-Lifshitz-Gilbert equation. The main idea of the construction is to inverse (in certain sense explained below) Hashimoto transform.
We will need the following auxiliary system:
\begin{definition}
Let $\mathbf{q}=q^1+iq^2\in L^{\infty}([0,\infty),L^2(\Omega\times S^1,\mathbb{C}))$, and 
$\mathbf{p}=p^1+ip^2\in L^{2}([0,\infty),L^2(\Omega\times S^1,\mathbb{C}))$.

\begin{equation}
d_t\begin{pmatrix}
\mathbf{u}\\
\mathbf{e}\\
\mathbf{u}\times\mathbf{e}
\end{pmatrix}
=
\begin{pmatrix}
0 & p^1 & p^2\\
-p^1 & 0 & C(\mathbf{q})\\
 -p^2 & -C(\mathbf{q}) & 0
\end{pmatrix}
\cdot
\begin{pmatrix}
\mathbf{u}\\
\mathbf{e}\\
\mathbf{u}\times\mathbf{e}
\end{pmatrix}
dt+
\begin{pmatrix}
0 & dW^1 & dW^2\\
-dW^1 & 0 & d\Psi\\
 -dW^2 & -d\Psi & 0
\end{pmatrix}
\circ
\begin{pmatrix}
\mathbf{u}\\
\mathbf{e}\\
\mathbf{u}\times\mathbf{e}
\end{pmatrix},\label{eqn:time_evolution}
\end{equation}
\begin{equation}
\partial_x\begin{pmatrix}
\mathbf{u}\\
\mathbf{e}\\
\mathbf{u}\times\mathbf{e}
\end{pmatrix}
=
\begin{pmatrix}
0 & q^1 & q^2\\
-q^1 & 0 & 0\\
-q^2 & 0 & 0
\end{pmatrix}
\cdot
\begin{pmatrix}
\mathbf{u}\\
\mathbf{e}\\
\mathbf{u}\times\mathbf{e}
\end{pmatrix},\label{eqn:space_evolution}
\end{equation}
where $W^1,W^2$ are independent Wiener processes given by
\begin{align}\label{eqn:noise_structure}
W^i(t,x)&=\sum\limits_{l=1}^{\infty}c^l\sigma^l(x)\beta^i_l(t),\sum\limits_{l=1}^{\infty}c_l^2<\infty,\\
&\{\beta_l^i\}_{l=1}^{\infty}-\mbox{ i.i.d. 1D Brownian motions}, \{\sigma^l\}_{l=1}^{\infty}-\mbox{ orthonormal basis in $L^2(S^1,\mathbb{R})$,}, i\in\mathbb{N}\nonumber
\end{align}
\begin{equation}\label{eqn:internal_structure}
\mathbf{p}=(\alpha+i\beta)\mathbf{q}_x, C(\mathbf{q})=-\frac12\beta|\mathbf{q}|^2+\frac{i\alpha}{2}\int\limits_{a}^x(\mathbf{q}_x\overline{\mathbf{q}}-\overline{\mathbf{q}}_x\mathbf{q})\,dy,
\Psi(t,x)=\int\limits_0^t\int\limits_a^x q^2\circ dW^1-q^1\circ dW^2,
\end{equation}
and stochastic integrals above are understood in Stratonovich sense.
\end{definition}
\begin{theorem}\label{thm:main}
Assume that we are given solution $\mathbf{q}\in L^2(\Omega, C([0,\infty),W^{2,2}(S^1,\mathbb{C})))$ of the following equation
\begin{align}
d_t\mathbf{q} &= \left((\alpha+i\beta)\mathbf{q}_{xx}+i\beta\frac{\mathbf{q}|\mathbf{q}|^2}{2}+\frac{\alpha\mathbf{q}}{2}\int\limits_{a}^x(\mathbf{q}_x\overline{\mathbf{q}}
-\overline{\mathbf{q}}_x\mathbf{q})\,dy\right)\,dt\nonumber\\
&\quad+d\partial_x(W^1+iW^2)-i\mathbf{q}\int\limits_a^x q^2\circ d_tW^1-q^1\circ d_tW^2,\label{eqn:Nonlinear_Stoch_Heat}
\end{align}
and $\mathbf{m}\in \mathbb{S}^2$, $\mathbf{e}_0\in T_m\mathbb{S}^2, |\mathbf{e}_0|=1$. 
Then system \eqref{eqn:time_evolution}, \eqref{eqn:space_evolution}, \eqref{eqn:noise_structure},
\eqref{eqn:internal_structure} has a solution $(\mathbf{u},\mathbf{e})\in L^{\infty}([0,\infty),L^2(\Omega,H^1(S^1,TS^2)))$ such that $\mathbf{u}(0,a)=\mathbf{m},\mathbf{e}(0,a)=\mathbf{e}_0$. Furthermore,
$\mathbf{u}$ solves Landau-Lifshitz-Gilbert equation
\begin{equation}
d \mathbf{u}=(\beta\mathbf{u}\times\mathbf{u}_{xx}-
\alpha\mathbf{u}\times(\mathbf{u}\times\mathbf{u}_{xx}))\,dt+\mathbf{u}\times \circ d\widetilde{\bf W}_t\label{eqn:SLLGtarget}
\end{equation}
where 
\[
\widetilde{\bf W}_t=\int\limits_0^t\mathbf{e}(s)dW^2(s)+\mathbf{e}\times\mathbf{u}dW^1(s)+
\mathbf{u}(s)dW^3(s),
\]
gaussian process with zero mean and quadratic covariation given by formula
\begin{align}
\mathbb{E}&\left[<\widetilde{\bf W}_t,\phi>_{L^2(S^1,\mathbb{R}^3)}<\widetilde{\bf W}_t,\psi>_{L^2(S^1,\mathbb{R}^3)}\right]\label{eqn:Quadvariation}\\
&=\sum\limits_{l=1}^{\infty}c_l^2\int\limits_0^t\mathbb{E}\Big[\int\limits_{S^1}<\phi,\mathbf{e}>_{\mathbb{R}^3}\sigma^l\,dx\int\limits_{S^1}<\psi,\mathbf{e}>_{\mathbb{R}^3}\sigma^l\,dx+\int\limits_{S^1}<\phi,\mathbf{u}>_{\mathbb{R}^3}\sigma^l\,dx\int\limits_{S^1}<\psi,\mathbf{u}>_{\mathbb{R}^3}\sigma^l\,dx\nonumber\\
&+\int\limits_{S^1}<\phi,\mathbf{u}\times\mathbf{e}>_{\mathbb{R}^3}\sigma^l\,dx\int\limits_{S^1}<\psi,\mathbf{u}\times\mathbf{e}>_{\mathbb{R}^3}\sigma^l\,dx
\Big]\,ds,\phi,\psi\in L^2(S^1,\mathbb{R}^3).\nonumber
\end{align}
Moreover, $\mathbf{u}(t)$ and $\mathbf{q}(t),t\geq 0$ are connected with each other through Hashimoto transform introduced in the previous theorem.
\end{theorem}
\begin{remark}\label{rem:Space_time_White_noise}
Note that if $c_l=1,l\in\mathbb{N}$ then from formula \eqref{eqn:Quadvariation} follows that
\[
\mathbb{E}\left[<\widetilde{\bf W}_t,\phi>_{L^2(S^1,\mathbb{R}^3)}<\widetilde{\bf W}_t,\psi>_{L^2(S^1,\mathbb{R}^3)}\right]=t<\phi,\psi>_{L^2(S^1,\mathbb{R}^3)},\phi,\psi\in L^2(S^1,\mathbb{R}^3),
\]
i.e. if $W^i, i=1,2$ are two independent real-valued space-time white noises then $\widetilde{\bf W}$ is an $\mathbb{R}^3$ valued space-time white noise.
\end{remark}
\begin{proof}[Proof]
The system \eqref{eqn:time_evolution}, \eqref{eqn:space_evolution}, \eqref{eqn:noise_structure}, \eqref{eqn:internal_structure} has a solution iff and only if compatibility conditions
\begin{align}
d_t\partial_x \mathbf{u} &= \partial_x d_t\mathbf{u},\label{eqn:uCompatibility}\\
d_t\partial_x \mathbf{e} &= \partial_x d_t\mathbf{e},\label{eqn:eCompatibility}
\end{align}
are satisfied. First, we will look at the condition \eqref{eqn:uCompatibility}.
We have by elementary calculations and equality \eqref{eqn:space_evolution} that
\begin{align}
d_t \partial_x\mathbf{u} &= (d_t q^1-q^2 C(\mathbf{q})-d\Psi)\mathbf{e}-(p^1q^1+p^2q^2+q^1dW^1+q^2dW^2)\mathbf{u}\label{eqn:txu_der}\\
&+(d_t q^2+q^1 C(\mathbf{q})+d\Psi)\mathbf{u}\times\mathbf{e},\nonumber
\end{align}
\begin{align}
\partial_x d_t\mathbf{u} &= (\partial_x p^1+d_t\partial_x W^1)\mathbf{e}-
(p^1q^1+p^2q^2+q^1dW^1+q^2dW^2)\mathbf{u}\label{eqn:xtu_der}\\
&+(\partial_x p^2+d_t\partial_x W^2)\mathbf{u}\times\mathbf{e}.\nonumber 
\end{align}
Equating coefficients in \eqref{eqn:txu_der} and \eqref{eqn:xtu_der} we can deduce that
\begin{equation}
d_t\mathbf{q}=\left[-i\mathbf{q}C(\mathbf{q})+\partial_x\mathbf{p}\right]\,dt-i\mathbf{q}d\Psi
+d\partial_x(W^1+i W^2).\label{eqn:uCompatibility_1}
\end{equation}
Now let us look at the second compatibility condition \eqref{eqn:eCompatibility}.
We have by elementary calculations and equality \eqref{eqn:space_evolution} that
\begin{equation}
d_t \partial_x\mathbf{e} = - d_t q^1\mathbf{u}- (q^1p^1 dt + dW^1)\mathbf{e}-
(q^1p^2 +q^1dW^2) \mathbf{u}\times\mathbf{e},\label{eqn:txe_der}
\end{equation}
and
\begin{align}\label{eqn:xte_der}
\partial_x d_t\mathbf{e} &= ([-\partial_x p^1-q^2C(q)]\,dt-d_t\partial_x W^1-q^2d\Psi)\mathbf{u}
-(p^1q^1+q^1dW^1)\mathbf{e}\\
&+([-p^1q^2 + \partial_x C(\mathbf{q})]\,dt-q^2 dW^1+d\partial_x\Psi)\mathbf{u}\times\mathbf{e}.\nonumber
\end{align}
Equating coefficients in \eqref{eqn:txe_der} and \eqref{eqn:xte_der} we can deduce that
\begin{equation}\label{eqn:eCompatibility_1}
[p^1q^2-p^2q^1-\partial_x C(\mathbf{q})]\,dt+q^2 dW^1-q^1 dW^2-d\partial_x\Psi = 0.
\end{equation}
Now we can notice that compatibility conditions \eqref{eqn:uCompatibility_1} and \eqref{eqn:eCompatibility_1} together with \eqref{eqn:internal_structure} give us equation
\eqref{eqn:Nonlinear_Stoch_Heat}. It remains to show that equation \eqref{eqn:SLLGtarget} holds.
We have by equation \eqref{eqn:time_evolution} that
\begin{equation}
d_t \mathbf{u} = (p^1\mathbf{e}+ p^2\mathbf{u}\times\mathbf{e})\,dt+ dW^1\mathbf{e}+dW^2\mathbf{u}\times\mathbf{e}.\label{eqn:SLLGidentify_1}
\end{equation}
Moreover, we can deduce by elementary calculation and equation \eqref{eqn:space_evolution} that
\begin{equation*}
\partial_{xx}^2 \mathbf{u}=q^1_x\mathbf{e}+q_x^2\mathbf{u}\times\mathbf{e}-|\mathbf{q}|^2\mathbf{u},
\end{equation*}
and, consequently,
\begin{equation*}
\beta \mathbf{u}\times \partial_{xx}^2 \mathbf{u}-\alpha \mathbf{u}\times(\mathbf{u}\times\partial_{xx}^2 \mathbf{u})=
(\alpha q_x^1-\beta q_x^2)\mathbf{e}+(\alpha q_x^2+\beta q_x^1)\mathbf{u}\times\mathbf{e}
\end{equation*}
Hence identity \eqref{eqn:internal_structure} $\mathbf{p}=(\alpha+i\beta)\mathbf{q}_x$ implies that
\begin{equation}
\beta \mathbf{u}\times \partial_{xx}^2 \mathbf{u}-\alpha \mathbf{u}\times(\mathbf{u}\times\partial_{xx}^2 \mathbf{u})=
p^1\mathbf{e}+p^2\mathbf{u}\times\mathbf{e}.\label{eqn:SLLGidentify_2}
\end{equation}
Combining \eqref{eqn:SLLGidentify_1} and \eqref{eqn:SLLGidentify_2} we deduce equation 
\eqref{eqn:SLLGtarget}. It remains to show that $\mathbf{u}$ and $\mathbf{q}$ are connected through Hashimoto transform. We will deduce it from the equation \eqref{eqn:space_evolution}.
First, by antisymmetry of the matrices in the equations \eqref{eqn:space_evolution} and \eqref{eqn:time_evolution} we can see that quanities $|\mathbf{u}|_{\mathbf{R}^3}$, $|\mathbf{e}|_{\mathbf{R}^3}$, $<\mathbf{u},\mathbf{e}>_{\mathbf{R}^3}$ are constants both in space and time. 
Consequently, by conditions on $\mathbf{u}_0$ and $\mathbf{e}_0$ we have that
$|\mathbf{u}|_{\mathbf{R}^3}=|\mathbf{e}|_{\mathbf{R}^3}=1$, $<\mathbf{u},\mathbf{e}>_{\mathbf{R}^3}=0$. We have system
\begin{align}\label{eqn:space_evo_1}
\partial_x\mathbf{u} &= q^1 \mathbf{e}+ q^2\mathbf{u}\times\mathbf{e},\\
\partial_x\mathbf{e} &= -q^1 \mathbf{u}.\label{eqn:space_evo_2}
\end{align}
Equation \eqref{eqn:space_evo_1} immediately implies that $|\mathbf{q}|=|\partial_x\mathbf{u}|$.
Consequently, we have
\begin{equation}\label{eqn:q_u_aux_1}
q^1=|\partial_x\mathbf{u}|\cos\omega,\,\, q^2=|\partial_x\mathbf{u}|\sin\omega.
\end{equation}
Moreover, elementary calculations allow us to deduce from equations \eqref{eqn:space_evo_1}-\eqref{eqn:space_evo_2} that 
\begin{align}\label{eqn:e_representation}
\mathbf{e} &= \frac{q^1 \partial_x\mathbf{u}+ q^2\mathbf{u}\times \partial_x\mathbf{u}}{|\partial_x\mathbf{u}|^2}\\
&= \frac{\cos\omega \partial_x\mathbf{u}+ \sin\omega\mathbf{u}\times \partial_x\mathbf{u}}{|\partial_x\mathbf{u}|}.\label{eqn:e_representation_2}
\end{align}
Furthermore, equation \eqref{eqn:space_evo_2} implies that 
\begin{equation}\label{eqn:q_u_aux_2}
<\partial_x\mathbf{e},\mathbf{u}\times\mathbf{e}>_{\mathbb{R}^3}=0.
\end{equation}
Inserting in the equation \eqref{eqn:q_u_aux_2} representation \eqref{eqn:e_representation_2} of $\mathbf{e}$ we deduce that 
\begin{equation}\label{eqn:q_u_aux_3}
\omega_x=\frac{<\mathbf{u}\times\partial_x\mathbf{u},\partial_{xx}^2\mathbf{u}>}{|\partial_x\mathbf{u}|^2}.
\end{equation}
Equation \eqref{eqn:q_u_aux_3} together with representation \eqref{eqn:q_u_aux_1} implies the result.
\end{proof}
\begin{remark}
The equation \eqref{eqn:Nonlinear_Stoch_Heat} is a zero-curvature representation (although no dependence on spectral parameter here) of the system \eqref{eqn:time_evolution}, \eqref{eqn:space_evolution}, \eqref{eqn:noise_structure}, \eqref{eqn:internal_structure}. Consequently, natural question is if there exist soliton solutions for LLG and SLLG equations?
\end{remark}
\begin{remark}
The existence of the strong solution of the equation \eqref{eqn:Nonlinear_Stoch_Heat} is a subject of further research.
\end{remark}
\begin{remark}
The existence and regularity assumptions on the solution of the equation \eqref{eqn:Nonlinear_Stoch_Heat} in the theorem \ref{thm:main} can be circumvented 
by consideration of proper space discretisation of the system \eqref{eqn:time_evolution}-\eqref{eqn:space_evolution}.
\end{remark}
\begin{remark}
Here we consider only the case of exchange energy $\mathcal{E}(u):=\int\limits_{\mathbb{S}^1}|\mathbf{u}_x|^2\,dx$ in the SLLG equation. The general case can be obtained in the same fashion
by properly modifying definitions of $\mathbf{p}$ and $C(\mathbf{q})$ in \eqref{eqn:internal_structure}.
\end{remark}
\begin{remark}
System \eqref{eqn:time_evolution}, \eqref{eqn:space_evolution}, \eqref{eqn:internal_structure} in the case of absence of noise ($W^1=W^2=\Psi=0$) and zero viscosity coefficient has been considered in \cite{NSVZ2007}.
\end{remark}
\begin{remark}
It would be of interest to see if in the case of absence of the noise term the equation \eqref{eqn:Nonlinear_Stoch_Heat} is equivalent to the complex Ginzburg-Landau equation deduced in \cite{Melcher2012}.
\end{remark}
\paragraph*{Acknowledgements:} I would like to thank Ben Goldys, Zdzislaw Brze\'{z}niak and Andreas Prohl for useful comments and attention to the work. The support by the ARC Discovery grant DP120101886 is gratefully acknowledged.




\begin{thebibliography}{99}

\bibitem{BBNP2013} L.~Banas, Z.~Brze\'{z}niak, M.~Neklyudov, A.~Prohl, {\em Stochastic Ferromagnetism -- Analysis and Numerics}, De Gruyter Studies in Mathematics,~{\bf 58}, 2013.  

\bibitem{BGJ2013} Z.~Brze\'{z}niak, B.~Goldys, T.~Jegaraj, {\em Weak solutions of a stochastic Landau-Lifshitz-Gilbert equation}, Appl. Math. Res. Express. AMRX, no. 1, pp.~1--33 (2013).

\bibitem{FT2007}  L.~D.~Faddeev, L.~A.~Takhtajan, {\em Hamiltonian methods in the theory of solitons},  Classics in Mathematics, Springer, Berlin, 2007. 

\bibitem{H1972} H.~Hasimoto, {\em A soliton on a vortex filament}, J. of Fluid Mech.~{\bf 51}, no. 3, pp.~477--485 (1972).

\bibitem{Melcher2012} C.~Melcher, {\em Global solvability of the Cauchy problem for the Landau-Lifshitz-Gilbert equation in higher dimensions}, Indiana Univ. Math. J.~{\bf 61}, no. 3, pp.~1175–-1200 (2012).

\bibitem{NSVZ2007} A.~Nahmod,  J.~Shatah, L.~Vega, C.~Zeng,
{\em Schrödinger maps and their associated frame systems}, 
Int. Math. Res. Not., no. {\bf 21}, Art. ID rnm088, 29 pp (2007).

\end{thebibliography}
\end{document}